\documentclass[journal,doublecolumn,10pt]{IEEEtran}

\usepackage{epsfig,latexsym}
\usepackage{float}
\usepackage{indentfirst}
\usepackage{amsmath}
\usepackage{amssymb}
\usepackage{times}
\usepackage{subfigure}
\usepackage{psfrag}
\usepackage{cite}
\usepackage{lastpage}
\usepackage{fancyhdr}

\newtheorem{Definition}{Definition}

\newtheorem{Proposition}{Proposition}

\newtheorem{Lemma}{Lemma}

\newtheorem{theorem}{$\mathbf{Theorem}$}
\newtheorem{lemma}[Lemma]{$\mathbf{Lemma}$}
\newtheorem{proposition}[Proposition]{Proposition}

\begin{document}%
\title{    Power Allocation Strategies   in Energy Harvesting Wireless  Cooperative Networks}

\author{ Zhiguo Ding, \IEEEmembership{Member, IEEE},
Samir M. Perlaza, \IEEEmembership{Member, IEEE}, I\~{n}aki Esnaola, \IEEEmembership{Member, IEEE}, H. Vincent Poor, \IEEEmembership{Fellow, IEEE}\thanks{The material in this papers was presented in part
at the $8$-th International Conference on Communications and Networking in China,
Guilin, China, August, 2013. 

The authors are with Department of
Electrical Engineering, Princeton University, Princeton, NJ 08544,
USA.   Zhiguo Ding is also with  School of
Electrical, Electronic, and Computer Engineering Newcastle
University, NE1 7RU, UK.  }} \maketitle\vspace{-4em}
\begin{abstract}
  In this paper, a wireless  cooperative network is considered,  in which multiple source-destination pairs communicate with each other via an {\it energy harvesting} relay. The focus of this paper is on the relay's strategies to distribute the harvested energy among the multiple users and their impact on the system performance.   Specifically, a   non-cooperative  strategy is to use the energy harvested from the $i$-th source as the relay transmission power to  the $i$-th destination, to which   asymptotic results show that its outage performance decays as $\frac{\log SNR}{SNR}$. A faster decaying rate, $\frac{1}{SNR}$, can be achieved  by the two centralized strategies proposed this the paper, where the water filling based one can achieve optimal performance with respect to several criteria, with a price of high complexity.  An auction based power allocation scheme is also proposed to achieve a better tradeoff between the system performance and complexity. Simulation results are provided to confirm the accuracy of the developed analytical results and facilitate a better performance comparison.

\end{abstract}
\section{Introduction}
Low cost mobile devices have been recognized as  crucial components of various wireless networks with important applications. A typical example is   wireless sensor networks which have been developed for a variety of applications, including surveillance, environmental monitoring  and health care.  Such low cost devices are typically  equipped with fixed energy supplies, such as batteries with  limited operation life. Replacing batteries for such devices is either impossible or expensive, particularly in the case in which  sensors are deployed in hostile environments. Therefore energy harvesting, a technique to collect energy from the surrounding environment, has recently received considerable  attention as a sustainable   solution to overcome the bottleneck of energy constrained wireless networks \cite{Raghunathan06}.

Conventional energy harvesting techniques rely on external energy sources that are  not part of communication networks, such as   those  based on solar power, wind energy, etc. \cite{Raghunathan06,Paradiso05}. Recently a new concept of energy harvesting has been proposed which involves  collecting energy from ambient radio frequency signals \cite{varshney08,Grover10}, so that wireless signals can be used as a means  for the delivery of information and power simultaneously.  In addition, such an approach can also reduce the cost of communication networks, since peripheral equipment to take advantage of  external energy sources can be avoided.  The concept of simultaneous power and information delivery was first proposed in \cite{varshney08}, where the fundamental tradeoff between the energy and information rate   is characterized for point-to-point communication scenarios. The extension of such a concept to frequency selective channels is considered in \cite{Grover10}. In \cite{LiangLiu12} the authors  study energy harvesting for communication scenarios with co-channel interference, where such interference is identified as a potential energy source. The simultaneous transfer of power and information is also   studied in multiple-input multiple-output systems in \cite{Zhangr11g}, and its extension to the scenario with imperfect channel information at the transmitter was considered in \cite{Xiangzt12}.

To ensure such a new concept of energy harvesting implemented in practical systems,  it is important to address the difficulty that practical circuits cannot realize   energy  harvesting and data detection   from wireless signals at the same time. This challenge has motivated  a few recent works deviating  from the ideal assumption that a receiver can detect signals and harvest energy simultaneously.  In \cite{Zhouzhang13}, the authors introduced  a general receiver architecture, in which  the   circuits for energy harvesting and signal detection are operated in a time sharing or power splitting manner. This  approach is naturally applied to a cooperative network with one source-destination pair in \cite{Nasirzhou}, where amplify-and-forward  (AF) is considered and exact expressions for outage probability and throughput are developed.

In this paper,  a general wireless cooperative network is considered,  in which multiple pairs of sources and destinations communicate through an {\it energy harvesting} relay. Specifically,    multiple  sources deliver their information to the relay via orthogonal channels, such as different time slots. The relaying transmissions are  powered by the signals sent from the sources. Assuming   that the battery of the relay is sufficiently large, the relay can accumulate a large amount of power for   relaying transmissions. The aim of this paper is to  study how to efficiently distribute such power among the multiple users   and investigate the impact of these power allocation strategies   on the system performance.

The contribution of this paper is four-fold. {\it Firstly}, a non-cooperative  individual transmission strategy is developed, in which  the relaying transmission to the $i$-th destination is powered by only using the energy harvested from the $i$-th source. Such a simple power allocation scheme will   serve  as a benchmark for other more sophisticated strategies  developed  in the paper. The decode-and-forward  (DF) strategy is considered, and the exact expression of the outage probability achieved by such a scheme is obtained. Based on this expression, asymptotic studies are carried out to show that the average outage probability for such a scheme decays  with the signal-to-noise ratio (SNR) at a rate of $\frac{\log SNR}{SNR}$.

{\it Secondly}, the performance of an equal power allocation scheme is investigated, in which the relay   distributes the accumulated power harvested from the sources evenly among   relaying transmissions. The advantage of such a scheme is that a user pair with poor channel conditions can be helped since more relay transmission power will be allocated to them    compared to the individual transmission strategy.   Exact expressions for the outage performance achieved by this  transmission scheme are obtained. Analytical results show that the equal power allocation scheme can always outperform the individual transmission strategy. For example, the average outage probability achieved by the equal power allocation scheme decays at the rate of $\frac{1}{SNR}$, faster than the individual transmission scheme.

{\it Thirdly},  a more opportunistic power allocation strategy based on the sequential water filling principle is studied. The key idea of such a strategy is that the relay will  serve a user with a better channel condition first, and help  a user with a worse channel condition afterwards if there is any power left at the relay.  This sequential  water filling scheme can achieve the optimal performance for the user with the best channel conditions, and also maximize the number of successful destinations. Surprisingly it can also be  proved that such a scheme minimizes the worst user outage probability. Several bounds are developed for the average outage probability achieved by such a scheme, and asymptotic studies are carried out to show that such bounds exhibit the same   rate of decay at   high SNR.

{\it Finally}, an auction based power allocation scheme is proposed, and the property of its equilibrium is discussed. Recall that the sequential  water filling scheme can achieve  superior performance in terms of receptional reliability, however, such a scheme requires that   channel state information (CSI) is available at the transmitter, which can consume significant system overhead in a multi-user system. As demonstrated by the simulation results, the auction based distributed scheme can achieve much better performance than the equal power and individual transmission schemes,  and close to the water filling strategy.


 \section{Energy harvesting relaying transmissions}\label{section indi}
Consider an energy harvesting  communication scenario with $M$ source-destination pairs and one relay. Each node is equipped with a single antenna.  Each source communicates with its destination via the relay,  through orthogonal channels, such as different time slots. All channels are assumed to be quasi-static Rayleigh fading, and  large scale path loss will  be considered only in Section \ref{section numberical results} in order to simplify the analytical development.

The basic idea of energy harvesting relaying is that an energy constrained  relay   recharges its battery by using  the energy from its observations.  Among the various energy harvesting relaying models, we   focus on power splitting \cite{Zhouzhang13,Nasirzhou}. Specifically, the cooperative transmission consists of two time slots of  duration   $\frac{T}{2}$. At the end of the first phase,   the relay   splits the observations from the $i$-th transmitter into two streams, one for energy harvesting and the other for detection. Let $\theta_i$ denote  the power splitting coefficient for the $i$-th user pair, i.e. $\theta_i$ is the fraction of observations used for energy harvesting.  At the end of the first phase,    the relay's detection is based on the following observation
\begin{eqnarray}
y_{r,i} = \sqrt{(1-\theta_i)P_i} h_i s_i +n_{r,i} ,
\end{eqnarray}
where  $P_i$ denotes the transmission power at the $i$-th source,  $h_i$ denotes the channel gain between the $i$-th source and the relay, $s_i$ is the source message with unit power, and $n_{r,i}$ denotes   additive white Gaussian noise (AWGN) with unit variance. As discussed in \cite{Nasirzhou}, such noise consists of   the   baseband AWGN as well as the sampled AWGN due to the radio-frequency band to baseband signal conversion.   We consider a pessimistic case in which  power splitting only reduces the signal power, but not to the noise power, which can provide  a lower bound for  relaying networks in practice.

The data rate at which the relay can decode the $i$-th source's signal is
\begin{eqnarray}
R_{r,i} = \frac{1}{2} \log(1+(1-\theta_i)P_i|h_{i}|^2),
\end{eqnarray}
and the parameter $\theta_i$ can be set to satisfy the    criterion   $R_{r,i} =R $, i.e.,
\begin{eqnarray}
  \theta_i \triangleq  1-  \frac{2^{2R}-1}{P_i|h_{i}|^2},
\end{eqnarray}
where $R$ is the targeted data rate. The reason for the above choice of $\theta_i$ can be justified as follows. A larger value of $\theta_i$ yields more energy reserved for the second phase transmissions, and therefore is beneficial to improve the performance at the destination. On the other hand, a larger value of $\theta_i$ reduces the signal power for relay detection and hence degrades   the receptional reliability at the relay. For that reason, a reasonable  choice is to use a    $\theta_i$ that assumes    successful detection at the relay, i.e.,  $R_{r,i} \geq R $.

 At the end of the first phase, the relay   harvests the following amount of energy from the $i$-th source:
\begin{eqnarray}
E_{H,i} = \eta P_i|h_i|^2\theta_i \frac{T}{2},
\end{eqnarray}
where $\eta$ denotes the energy harvesting efficiency  factor. During the second time slot, this energy can be used to power the relay transmissions. However, how to best use  such harvested energy is not a trivial problem, since different strategies will have different impacts on  the system performance. In the following subsection, we first introduce a non-cooperative individual transmission strategy, which serves  as a benchmark for the transmission schemes proposed  later.

\subsection{A non-cooperative individual transmission strategy}
A straightforward strategy to use   the harvested energy is   allocating  the energy harvested from the $i$-th source to  the relaying transmission to the $i$-th destination, i.e.,   the relaying transmission power  for the $i$-th destination is
\begin{eqnarray}
P_{ri} \triangleq \frac{E_{H,i}}{\frac{T}{2}} = \eta P_i|h_i|^2\theta_i.
\end{eqnarray}
   During the second time slot, the DF relay   forwards the $i$-th source message if the message is reliably detected at the relay, i.e. $|h_i|^2>\epsilon$, where $\epsilon=\frac{2^{2R}-1}{P_i}$.
    Therefore,   provided that a successful detection at the relay, the $i$-th destination receives the observation,  $  \sqrt{P_{ri}} g_{i} s_i + n_{d,i}$,
which yields a data rate at the $i$-th destination of
\begin{eqnarray}
R_{d,i} = \frac{1}{2} \log(1+P_{ri}|g_{i}|^2),
\end{eqnarray}
where $g_i$ denotes the channel between the relay and the $i$-th destination and $n_{d,i}$ denotes the noise at the destination.   For notational simplicity, it is assumed that the noise at the destination  has the same variance as that  at the relay.  The   outage probability for the $i$-th   user pair can be  expressed as
\begin{eqnarray}\label{eq1}
\mathcal{P}_{i,I} = \mathrm{Pr}\left(\frac{1}{2}\log(1+P_i|h_{i}|^2)<R\right) + \mathrm{Pr}\left(\frac{1}{2}\log(1 \right. \\ \nonumber \left.+P_i|h_{i}|^2)>R, \frac{1}{2} \log(1+P_{ri}|g_{i}|^2) <R \right).
\end{eqnarray}
The following proposition characterizes  the outage of such a strategy.
\begin{proposition}\label{lemma 1}
The use of the non-cooperative  individual transmission strategy yields an outage probability at the $i$-th   destination of
\begin{eqnarray}
\mathcal{P}_{i,I} =  1 - \frac{e^{-\frac{a}{P_i}}}{P_i} \sqrt{\frac{4aP_i}{\eta}}\mathbf{K}_1\left(\sqrt{\frac{4a}{\eta P_i}}\right),
\end{eqnarray}
where $a=2^{2R}-1$ and $\mathbf{K}_n(\cdot)$ denotes the modified Bessel function of the second kind with order $n$. The worst and best outage performance among the $M$ users are $ 1-\left(1-\mathcal{P}_{i,I}\right)^M $ and $ \left(\mathcal{P}_{i,I}\right)^M$, respectively.
\end{proposition}
\begin{proof}
The first term on the righthand side of \eqref{eq1} can be   calculated as  $\left(1-e^{-\frac{a}{P_i}}\right)$ by using the exponential distribution. On denoting  the second probability on  the righthand side of \eqref{eq1} by $Q_2$, we have
 \begin{eqnarray}
Q_2 &=&  \mathrm{Pr}\left(\frac{1}{2}\log(1+P_i|h_{i}|^2)>R, \frac{1}{2} \log(1\right.  \\ \nonumber &&\left.+\eta P_i |h_{i}|^2\theta_i |g_{i}|^2) <R \right)\\ \nonumber &=&
\mathrm{Pr}\left(\frac{1}{2}\log(1+P_i|h_{i}|^2)>R, \frac{1}{2} \log(1\right. \\ \nonumber && \left.+\eta\left(  P_i |h_{i}|^2-  2^{2R}+1 \right) |g_{i}|^2) <R \right).
\end{eqnarray}
On setting  $z=P_i |h_{i}|^2-  2^{2R}+1 $,   we can write  the density function of $z$ as $f_z(z) = \frac{1}{P_i}e^{-\frac{z+a}{P_i}}$, which yields
\begin{eqnarray}
Q_2 &=&
\int^{\infty}_{0}\left(1-e^{-\frac{a}{\eta z}}\right) f_z(z)dz \\ \nonumber &=& e^{-\frac{a}{P_i}} - \frac{e^{-\frac{a}{P_i}}}{P_i} \sqrt{\frac{4aP_i}{\eta}}\mathbf{K}_1\left(\sqrt{\frac{4a}{\eta P_i}}\right),
\end{eqnarray}
where the last equation is obtained by applying Eq. (3.324.1) in \cite{GRADSHTEYN}. Combining the two probabilities in \eqref{eq1}, the first part of the lemma is proved. The worst and best outage performance can be obtained by using the assumption that  the  channels are   identically and independently distributed.
\end{proof}
   The asymptotic high SNR  behavior  of the outage performance can be used as an benchmark for comparing     power allocation strategies. Our intuition  is that such a straightforward strategy is most likely inefficient, as illustrated in the following. Suppose that  two source nodes with channels $|h_1|^2>> |h_2|^2\approx \frac{a}{P_2}$ and $|g_1|^2>> |g_2|^2$ have information correctly detected  at the relay. Based on the individual transmission scheme, there is little energy harvested from the second source transmission, which results in  $P_{r2}\rightarrow 0$ and therefore a possible detection failure at the second destination.   A more efficient solution to such a case is to allow the users to share the harvested power efficiently, which can  help the user with a poor connection. This scenario is discussed in the following sections.

\section{Centralized mechanisms for power allocation}\label{section centralized}
Recall that each user uses the power splitting fraction $ \theta_i= 1-  \frac{2^{2R}-1}{P_i|h_{i}|^2}$ , which implies that total power reserved at the relay at the end of the first phase is\footnote{Instead of all $M$ sources, we  consider only the power harvested from the $N$ sources which can deliver their information  to the relay successfully. Or in other words, we consider a pessimistic strategy that for each source, the relay will first direct the received signals  to the detection circuit until the receive SNR is sufficient for successful detection. If there is a failure of detection, all the energy must have already been directed to the detection circuit, and there will be no energy left for energy harvesting.    }
\begin{eqnarray}
P_r = \sum^{N}_{i=1}\frac{ E_{H,i}}{\frac{T}{2}} =  \sum^{N}_{i=1}  \eta P_s |h_{i}|^2\theta_i,
\end{eqnarray}
where $N$ denotes the number of   sources whose information can be reliably detected at the relay. Note that $N$ is a random variable whose value depends on the instantaneous source-relay channel realizations.
 To simplify the analysis, it is assumed that all the source transmission powers are the same $P_i=P_s$.
In the following, we   study how to distribute such power among the users based on various  criteria. Specifically, an equal power allocation strategy is introduced first, and then we will investigate the water filling based strategy which achieves a better outage performance but requires more complexity.

\subsection{Equal power allocation}
In this strategy, the relay   allocates the same amount of power to each user, i.e., $P_{ri}= \frac{1}{N} \sum^{N}_{i=1}  \eta P_s |h_{i}|^2\theta_i$. The advantage of such a strategy is that there is no need for the relay to know the relay-destination channel information, which can reduce the system overhead significantly, particularly in a multi-user system. The following theorem   describes  the outage performance achieved by such a power allocation scheme.
 \begin{theorem}\label{theorem1}
  Based on the equal power allocation, the outage probability for the $i$-th   destination is   given by
    \begin{eqnarray}\nonumber
\mathcal{P}_{i,II} =\nonumber
\sum^{M}_{n=1} \frac{1}{(n-1)!}\left((n-1)! -  2\left(\frac{b_n}{P}\right)^{\frac{n}{2}}\mathbf{K}_{n}\left(2\sqrt{\frac{b_n}{P}}\right) \right) \\   \times\frac{(M-1)!}{(n-1)!(M-n)!} e^{-n\epsilon}\left(1-e^{-\epsilon}\right)^{M-n}+\left(1-e^{-\epsilon}\right) ,\nonumber
\end{eqnarray}
where   $b_n=\frac{na}{\eta}$.
 \end{theorem}
 \begin{proof}
 See the appendix.
\end{proof}

Based on Theorem \ref{theorem1}, we also obtain the best outage and worst outage performance among the $M$ users achieved by the equal power allocation scheme as follows.
\begin{proposition}\label{corollary}
Based on the use of the equal power allocation, the outage probability  of the user with the best channel conditions among the $M$ users is
  \begin{eqnarray}
\mathcal{P}_{best,II} =\nonumber
   \sum^{M}_{n=1} \frac{2}{(n-1)!}\sum^{n}_{i=0}{n \choose i} (-1)^i \left(\frac{ib_n}{P}\right)^{\frac{n}{2}} \\ \nonumber \times\mathbf{K}_{n}\left(2\sqrt{\frac{ib_n}{P}}  \right) \frac{M!e^{-n\epsilon}}{n!(M-n)!} \left(1-e^{-\epsilon}\right)^{M-n} +\left(1-e^{-\epsilon}\right)^M,
\end{eqnarray}
 and the worst outage performance among the $M$ users is
    \begin{eqnarray}
\mathcal{P}_{worst,II} =\nonumber
 \frac{1}{(M-1)!}\left((M-1)! -  2\left(\frac{Mb_M}{P}\right)^{\frac{M}{2}}\right. \\ \nonumber \left. \times\mathbf{K}_{M}\left(2\sqrt{\frac{Mb_M}{P}}\right) \right) e^{-M\epsilon} +1-e^{-M\epsilon}.
\end{eqnarray}
\end{proposition}
\begin{proof}
Suppose that  there are $n$ sources whose messages can be reliably received by the relay. Among these $n$ users, order the relay-destination channels as
$
 g_{(1)}\leq \cdots\leq  g_{(n)},
 $
 and  the outage performance for the best outage performance can be expressed as
 \begin{eqnarray}
\nonumber
\mathcal{P}_{best,II} &=&\sum^{M}_{n=1} \mathrm{Pr}\left( \frac{1}{2}\log\left(1+\frac{P_r}{N}|g_{(N)}|^2\right)<R, N=n\right)\\ \nonumber &&+\mathrm{Pr}\left(N=0\right)\\ \nonumber
&=&\sum^{M}_{n=1} \mathrm{Pr}\left( \frac{1}{2}\log\left(1+\frac{P_r}{N}|g_{(N)}|^2\right)<R| N=n\right)\\  \label{eq bes} &&\times \mathrm{Pr}(N=n)+\mathrm{Pr}\left(N=0\right).
\end{eqnarray}
By applying the density function of $\sum^{N}_{i=1}| {h}_{n_i}|^2$ shown in the proof for Theorem \ref{theorem1}, the best outage probability can be expressed as \begin{eqnarray}\nonumber
\mathcal{P}_{best,II}  &=&  \sum^{M}_{n=1}\mathrm{Pr}(N=n)\cdot\int_{n\epsilon } ^{\infty} \left(1- e^{-\frac{b_n}{Py-na}}\right)^n \\ \nonumber&& \times \frac{1}{(n-1)!} (y-n\epsilon)^{n-1} e^{-(y-\epsilon)}dy  +\mathrm{Pr}\left(N=0\right).
\end{eqnarray}
By applying the binomial coefficients and \eqref{eq32}, the best outage probability can be obtained as shown in the proposition. The worst outage probability can be expressed as
 \begin{eqnarray}\nonumber
\mathcal{P}_{worst,II} &=&  \mathrm{Pr}\left( \frac{1}{2}\log\left(1+\frac{P_r}{M}|g_{(1)}|^2\right)<R, N=n\right)\\   &&+\sum^{M-1}_{n=0}\mathrm{Pr}\left(N=n\right).
\end{eqnarray}
Note that
 \begin{eqnarray}\label{eq 4}
\sum^{M-1}_{n=1}\mathrm{Pr}\left(N=n\right) = 1- \mathrm{Pr}(N=M)=1-e^{-M\epsilon}.
\end{eqnarray}
 Combining  the density function of $\sum^{N}_{i=1}| {h}_{n_i}|^2$ shown in the proof for Theorem \ref{theorem1}, and the results in \eqref{eq32} and \eqref{eq 4},   the  probability can be evaluated   and the proposition is proved.
 \end{proof}

\subsection{Sequential water filling based power allocation  strategy}
Provided that the relay has   access to   global channel state information, a more efficient strategy that    maximizes the number of successful destinations can be designed as follows.
First recall that in order to ensure the successful detection at the $i$-th destination, the relay needs to allocate the relaying transmission power $P_{i, targeted} = \frac{2^{2R}-1}{|g_i|^2}$ to the $i$-th destination. Suppose    that $n$ sources can deliver their information to the relay reliably, and the required relaying transmission power for these $n$ destinations can be ordered as
\[
  \frac{2^{2R}-1}{|g_{(1)}|^2} \geq \dots \geq  \frac{2^{2R}-1}{|g_{(n)}|^2}.
\]
The sequential water filling power allocation strategy is described in the following. The relay first serves the destination with the strongest channel   by allocating    power $\frac{2^{2R}-1}{|g_{(n)}|^2}$ to it, if the total harvested energy at the relay is larger than or equal to $\frac{2^{2R}-1}{|g_{(n)}|^2}$ . And then the relay  tries to serve the destination with the second strongest channel with the power $\frac{2^{2R}-1}{|g_{(n-1)}|^2}$, if possible.  Such a power allocation strategy continues until either all users are served or there is not enough   power left at the relay. If there is any power left, such energy is reserved at the relay, where it is assumed that the capacity of the relay battery is infinite.

The probability of having  $m$ successful receivers among $n$ users can be expressed as
\begin{eqnarray}\nonumber
  \mathrm{Pr}\left( \sum^{m}_{i=1} \frac{2^{2R}-1}{|g_{(n-i+1)|^2}}<P_r, \sum^{m+1}_{i=1} \frac{2^{2R}-1}{|g_{(n-i+1)}|^2}>P_r , N=n\right),
\end{eqnarray}
from which  the averaged number of   successful destinations can be calculated by carrying out the summation among all  possible choices of $m$ and $n$.
Evaluating  the above expression is quite challenging, mainly because of the complexity of the density function of the sum of inverse exponential variables. However, explicit    analytical results for  such a power allocation scheme can be obtained based on other criteria. Particularly we are interested in the outage performance achieved by the water filling  strategy.

Although such a water filling  power allocation scheme is designed to maximize the number of the successful destinations, it can also minimize the outage probability for the user with the best channel conditions, since such a user   is the first to be served and has the access to the maximal relaying power. The following proposition provides an explicit expression of such a outage probability.
\begin{proposition}\label{proposition x}
With the sequential water filling power allocation strategy, the outage probability for the user with the best channel conditions is
\begin{eqnarray}
\mathcal{P}_{best,III} =\nonumber
\sum^{M}_{n=1} \frac{2}{(n-1)!}\sum^{n}_{i=1}{n \choose i} (-1)^i   \left(\frac{i\tilde{b}}{P}\right)^{\frac{n}{2}}\mathbf{K}_{n}\left(2\sqrt{\frac{i\tilde{b}}{P}}\right)  \\ \nonumber \times\frac{M!}{n!(M-n)!} e^{-n\epsilon}\left(1-e^{-\epsilon}\right)^{M-n} +\left(1-e^{-\epsilon}\right)^M ,
\end{eqnarray}
where $\tilde{b}=\frac{a}{\eta}$.
\end{proposition}
\begin{proof}
 The outage probability for the user with the best channel channel conditions is
\begin{eqnarray}
\label{eq80}
\mathcal{P}_{best,III} =  \sum^{M}_{n=1} \mathrm{Pr}\left( \sum^{n}_{i=1}\eta \left(P_s|h_{i}|^2-  a \right) <    \frac{2^{2R}-1}{|g_{(n)}|^2} ,\right.\\\nonumber \left.N=n\right) + \mathrm{Pr}(N=0).
\end{eqnarray}
Following   steps similar to those steps in the proofs for Theorem \ref{theorem1} and Proposition \ref{corollary},   the   probability in Proposition \ref{proposition x} can be evaluated, and the details of such  algebraic  manipulations are omitted.
\end{proof}

The optimality of the water filling scheme in terms of the number of successful destinations and   the   performance for the user with the best channel conditions is straightforward to demonstrate. However, it is surprising   that the performance of the water filling scheme for the user with the worst outage probability is the same as that attained for the worst user with the optimal strategy,
as shown in the following lemma.
\begin{lemma}\label{lemma max min}
Denote by  $\mathcal{P}_i(\sf{s})$   the outage probability for the $i$-th user achieved by a power allocation strategy $\sf{s}$, where $\sf{s}\in \mathcal{S}$ and $\mathcal{S}$ contains all possible strategies. Define $\mathcal{P}_{worst}(\sf{s}) \triangleq \max\{\mathcal{P}_i(\sf{s}), i\in\{1, \cdots, M\}\}$  and $\mathcal{P}_{worst, III}$ as the worst user performance achieved by the sequential water filling scheme.  $\mathcal{P}_{worst, III}=\min\{\mathcal{P}_{worst}(\sf{s}), \sf{s}\in\mathcal{S}\}$  holds.
\end{lemma}
\begin{proof}
See the appendix.
\end{proof}
Given the form in \eqref{eq8},  it is quite challenging to find   exact expression for such an outage probability, for  the following reason. Denote $z_i=\frac{1}{|g_i|^2}$. Since the channels are Rayleigh faded, the     probability density and cumulative distribution  functions of $z_i$ can be obtained as follows:
\begin{equation}\label{exxx}
f_{z_i}(z) = \frac{1}{z^2}e^{-\frac{1}{z}}, F_{z_i}(z) =  e^{-\frac{1}{z}}.
\end{equation}
Obtaining an  exact expression for  \eqref{eq8} requires the density function of $\sum^{M}_{i=1} z_i$, which is the sum of inverse exponential variables. The Laplace transform for the density function of an individual $z_i$ is $\mathcal{L}_{z_i}(s)=2\sqrt{s}\mathbf{K}_1(2\sqrt{s})$, so that  the Laplace transform for the overall sum is $\mathcal{L}_{\sum^{M}_{i=1} z_i }(s)$ is $2^M {s}^{\frac{M}{2}}\mathbf{K}^M_1(2\sqrt{s})$, a form difficult to invert.  There are a few existing results regarding to the sum of inverted Gamma/chi-square distributed variables \cite{Witkovsky01,Saleh07}; however, the case with $2$ degrees of freedom, i.e. inverse exponential variables, is still   an open problem, partly  due to the fact that its moments are not bounded.
The following proposition provides   upper and lower bounds of the outage performance of the users with the worst channel conditions.
\begin{proposition}\label{lemma worst}
The outage probability for the user with the worst channel conditions achieved by the water filling strategy can be upper bounded  by
\begin{eqnarray}\nonumber
\mathcal{P}_{worst, III}<e^{-M\epsilon} \int_{0}^{\infty}\left(1-       e^{-\frac{M^2}{w}}   \right. \\   -M\left.\int^{w}_{\frac{w}{M}}   \frac{ e^{-\frac{(M-1)^2}{w-v}-\frac{1}{v}}}{v^2} dv\right)f_w(w)dw  +1-e^{-M\epsilon} ,\label{upperbound 1}
\end{eqnarray}
and lower bounded by
\begin{eqnarray}\label{lowerbound1}
\mathcal{P}_{worst, III}> \left(1 -  \frac{2}{(M-1)!}  \left(\frac{M\epsilon }{\eta}\right)^{\frac{M}{2}}\mathbf{K}_{M}\left(2\sqrt{\frac{M \epsilon}{\eta}}\right)\right)\\ \nonumber \times e^{-M\epsilon} +1-e^{-M\epsilon},
\end{eqnarray}
where
 $f_w(w) = \frac{1}{(M-1)!}\left(\frac{\epsilon}{\eta}\right)^M w^{M-1} e^{-\frac{\epsilon}{\eta}w}$.
\end{proposition}
\begin{proof}
See the appendix.
\end{proof}
While the expression in \eqref{upperbound 1} can be evaluated by numerical methods, it is difficult to carry out asymptotic studies for such an expression with integrals, and the following proposition provides a bound slightly
looser than \eqref{upperbound 1}  that enables asymptotic analysis.
\begin{proposition}\label{lemma worst2}
The outage probability for the user with the worst channel conditions achieved by the water filling strategy can be upper bounded as follows:
\begin{align}\label{upperbound 2}
\mathcal{P}_{worst, III}&<1- e^{-M\epsilon}  \left(    \frac{2}{(M-1)!} \left(\frac{\epsilon M^2}{\eta}\right)^{\frac{M}{2}}\right. \\\nonumber   &\times \mathbf{K}_M\left(2\sqrt{\frac{\epsilon M^2}{\eta}}\right) +\left.\frac{M}{(M-1)!} \left(\frac{\epsilon}{\eta}\right)\right. \\ \nonumber &\times \left.\int^{M-1}_{c}    2  \left(\frac{a(y)\epsilon}{\eta}\right)^{\frac{M-1}{2}}\mathbf{K}_{M-1}\left(2\sqrt{\frac{a(y)\epsilon}{\eta}}\right)dy\right)    ,
\end{align}
where $a(y)= (y+1)\left(\frac{(M-1)^2+1}{y}\right)$ and $c$ is a constant to facilitate asymptotic analysis,  $c\in [0, M-1]$.
\end{proposition}
\begin{proof}
See the appendix.
\end{proof}
The upper bound in Proposition \ref{lemma worst} is a special case of the one in Proposition \ref{lemma worst2} by setting $c=0$ as shown in the appendix. The reason to use the parameter $c$ is to facilitate asymptotic analysis and ensure  that the factor $a(y)\epsilon$ approaches zero at high SNR, as illustrated  in the next section.

Recall that the two bounds in Proposition \ref{lemma worst} were developed based on \eqref{bounds used in worst}, which is recalled in the following:
\begin{eqnarray}
\mathcal{P}_w\left(   z_{(M)} >w \right) < \mathcal{P}_w\left(  \sum^{M}_{i=1} z_i >w \right) \\ \nonumber \leq \mathcal{P}_w\left(   z_{(M)} +(M-1)z_{(M-1)}>w \right),
\end{eqnarray}
where $z_i$ has been ordered as $z_{(1)}\leq \cdots \leq z_{(M)}$ and $w$ is a random variable related to the source-relay channels and transmission power. Intuitively  such bounds should be quite loose since the two order statistics, $z_{(M)}$ and $z_{(M-1)}$, are expected to  become the same with  large $M$.

However, as shown by the   simulation   in Section \ref{section numberical results}, such bounds are surprisingly   tight, even for large $M$. This is because  for the addressed scenario the statistical properties of $z_{(M)}$ and $z_{(M-1)}$ are very different. In the following  it will be shown that the expectations  of  $z_{(M-i)}$, $1\leq i \leq (M-1)$, \textit{are}, but the expectation of $z_{(M)}$ is \textit{not bounded}, for any fixed $M$.   The expectation of $z_{(M-1)}$ is
\begin{eqnarray}
\mathcal{E}_{z_{(M-1)}}&=& M(M-1)\int^{\infty}_{0} \frac{1}{x} e^{-\frac{M-1}{x}} \left(1-e^{-\frac{1}{x}}\right)dx\\ \nonumber &=& M(M-1) \sum^{\infty}_{k=1}\frac{(-1)^k}{k!} \int^{\infty}_{0} \frac{1}{x^{k+1}}e^{-\frac{M-1}{x}}dx,
\end{eqnarray}
where the first equality follows from the pdf of $z_{(M-1)}$, $f_{z_{(M-1)}}(x)=M(M-1)  \frac{e^{-\frac{M-1}{x}}}{x^2}  \left(1-e^{-\frac{1}{x}}\right)$, and the second equality follows from the series expansion of exponential functions. Furthermore we have
\begin{eqnarray}
\mathcal{E}_{z_{(M-1)}}&=&  M(M-1) \sum^{\infty}_{k=1}\frac{(-1)^k}{k!} \frac{1}{(M-1)^k} \\ \nonumber &<& M(M-1) \sum^{\infty}_{k=1}\frac{(-1)^k}{(M-1)^k} =(M-1)^2.
\end{eqnarray}
Therefore the moments of $ z_{(M-1)}$   are finite, but  the expectation of $z_{(M)}$ is not finite since
\begin{eqnarray}
\mathcal{E}_{z_{(M)}}= M\int^{\infty}_{0} \frac{1}{x} e^{-\frac{M}{x}}  dx = M \cdot \underset{t\rightarrow 0}{\lim}\mathbf{E}_i(-t)\rightarrow \infty.
\end{eqnarray}

%
\section{Asymptotic Analysis  of the Outage Performance }
In the previous sections,   exact expressions for the outage performance achieved by the addressed power allocation schemes have been developed.
Most of the these analytical results contain    Bessel functions, which makes it difficult to get any insight from the analytical results. In this section, high SNR asymptotic studies for the outage performance are  carried out.  To do this we need    asymptotic expression for $x^n\mathbf{K}_n(x)$, when $x\rightarrow 0$. By applying the series representation of Bessel functions,  $x^n\mathbf{K}_n(x)$ can be approximated as \cite{GRADSHTEYN}
\begin{eqnarray}\label{approximation}
x^n\mathbf{K}_n(x) &=& x^n(-1)^{n+1} \mathbf{I}_n(x)\left(\ln \frac{x}{2} +\mathbf{C}\right)\\ \nonumber && +\frac{1}{2}(-1)^n \sum^{\infty}_{l=0}\frac{\left(\frac{x}{2}\right)^{n+2l}x^n}{l!(n+l)!}\left( \sum^{l}_{k=1}\frac{1}{k} +\sum^{n+l}_{k=1}\frac{1}{k}  \right)\\ \nonumber &&+\frac{1}{2}\sum^{n-1}_{l=0}\frac{(-1)^l(n-l-1)!}{l!} \frac{x^{2l}}{2^{2l-n}} \\ \nonumber &\approx & \frac{1}{2}\sum^{n-1}_{l=0}\frac{(-1)^l(n-l-1)!}{l!} \frac{x^{2l}}{2^{2l-n}} +  x^{2n} q(\ln x),
\end{eqnarray}
for $n\geq 2$, where $\mathbf{I}_n(x)=\sum^{\infty}_{k=0} \frac{\left(\frac{x}{2}\right)^{n+2k}}{k! (n+k)!}$ and $q(\ln x)=  (-1)^{n+1}\frac{\ln \frac{x}{2} }{2^nn!}  $. For the case of $n=1$, we have
\begin{eqnarray}\label{approximation2}
x\mathbf{K}_1(x) &=&1+ x \mathbf{I}_1(x)\left(\ln \frac{x}{2} +\mathbf{C}\right) \\ \nonumber &&-\frac{1}{2} \sum^{\infty}_{l=0}\frac{\left(\frac{x}{2}\right)^{2l+1}x}{l!(l+2)!}\left( \sum^{l}_{k=1}\frac{1}{k} +\sum^{l+2}_{k=1}\frac{1}{k}  \right)\\ \nonumber &\approx & 1+ \frac{x^2}{2} \ln \frac{x}{2}   .
\end{eqnarray}
These approximations will be used for the following high SNR asymptotic analysis of the outage performance.
\subsection{Averaged outage performance }
According to Theorem \ref{theorem1},  the averaged performance achieved by the equal power allocation strategy can be expressed as
   \begin{eqnarray}\label{eq 222}
\mathcal{P}_{i,II}=\sum^{M}_{n=1} \frac{1}{(n-1)!}\left((n-1)! -  \frac{1}{2^{n-1}} \left(2\sqrt{\frac{b_n}{P}}\right) ^n \right. \\ \nonumber \left.\times \mathbf{K}_{n}\left(2\sqrt{\frac{b_n}{P}}\right) \right)  \frac{(M-1)!}{(n-1)!(M-n)!} e^{-n\epsilon}\\ \nonumber \times \left(1-e^{-\epsilon}\right)^{M-n} +\left(1-e^{-\epsilon}\right) \\ \nonumber \approx\nonumber
\sum^{M}_{n=1}  \left(1 -       \left( 1-\frac{1}{(n-1)} \left( {\frac{b_n}{P}}\right)  \right)
 \right) \frac{(M-1)!}{(n-1)!(M-n)!}\\ \nonumber\times   (1-n\epsilon) \epsilon^{M-n} + \epsilon  \\ \nonumber \approx
  \left(1+\frac{M}{(M-1)\eta}\right)\epsilon,
\end{eqnarray}
where only the first two factors containing $x^0$ and $x^2$ in \eqref{approximation} are used.
On the other hand, according to Proposition \ref{lemma 1}, the averaged outage performance achieved by the non-cooperative individual strategy is
\begin{eqnarray}\label{euq I}
\mathcal{P}_{i, I}
&\approx&  1 -  (1-\epsilon) \left( 1+  \frac{2\epsilon}{\eta} \ln  \sqrt{\frac{\epsilon}{\eta  }} \right) \\ \nonumber &\approx& \epsilon \left(1- \frac{2}{\eta} \ln  \sqrt{\frac{\epsilon}{\eta  }}\right).
\end{eqnarray}
An important observation from \eqref{eq 222} and \eqref{euq I} is that the averaged outage probability for the individual transmission scheme decays as $\frac{\log SNR}{SNR}$, where the equal power allocation scheme can achieve  better performance, i.e. a faster   rate of decay, $\frac{1}{SNR}$. Another aspect for comparison  is to study the normalized difference of the two probabilities. When $\epsilon\rightarrow 0$, we can approximate this difference as
\begin{eqnarray}
\frac{\mathcal{P}_{i, I}- \mathcal{P}_{i,II}}{\mathcal{P}_{i,II}}&\approx&   \frac{   \frac{1}{\eta} \ln   {\frac{1}{\epsilon\eta  }} -\frac{n}{2(n-1)\eta}  }{ 1+\frac{M}{(M-1)\eta}}>0.
\end{eqnarray}
This difference can be significant since the factor $\ln \frac{1}{\epsilon}$ approaches   infinity as $\epsilon \rightarrow 0$.  In terms of the averaged outage performance, the water filling strategy can also achieve  performance similar to that of the equal power allocation scheme, i.e., its averaged outage probability decays as $\epsilon$. Although we cannot obtain an explicit expression for the water filling strategy, the   rate of decay of $\frac{1}{SNR}$ can be  proved by studying  the outage probability  for the user with the worst channel conditions as shown in Section \ref{subsection worst}.

\subsection{Best outage performance}
Following the previous discussions about the averaged outage performance,  the best outage performance achieved by the individual transmission scheme can be approximated  as follows:
\begin{eqnarray}
\mathcal{P}_{best, I} = (\mathcal{P}_{i,I})^M \approx  \epsilon^M \left(1- \frac{2}{\eta} \ln  \sqrt{\frac{\epsilon}{\eta  }}\right)^M.
\end{eqnarray}
Comparing Proposition \ref{corollary} and Proposition \ref{proposition x}, we can see that  the equal power allocation scheme and the water filling scheme  achieve similar performance for the user with the best channel conditions. So in the following, we  focus only on the equal power allocation scheme. The following corollary provides a high SNR approximation of the best outage performance achieved by   equal power allocation.
\begin{proposition} \label{corollary 2}
With the equal power allocation scheme, the outage probability for the user with the best channel conditions can be approximated at high SNR by
  \begin{eqnarray}\nonumber
\mathcal{P}_{best,II}  &\approx&    \epsilon^M(1-c\ln \epsilon ),
\end{eqnarray}
where $c=\sum^{M}_{n=1} 2    \left(\frac{n}{\eta}\right)^n    \frac{M!}{(n-1)!n!(M-n)!} $ is a constant not depending  on $\epsilon$.
\end{proposition}
\begin{proof}
According to Proposition \ref{corollary},  the use of   equal power allocation yields the best outage performance among the $M$ users as in \eqref{approxi worst channel}.
\begin{figure*}
  \begin{eqnarray}
\mathcal{P}_{best,II} &=&\nonumber   2\sum^{M}_{n=1} \frac{1}{(n-1)!}\sum^{n}_{i=0}{n \choose i} (-1)^i \left(\frac{ib}{P}\right)^{\frac{n}{2}}\mathbf{K}_{n}\left(2\sqrt{\frac{ib}{P}}  \right)  \frac{M!}{n!(M-n)!} e^{-n\epsilon}\left(1-e^{-\epsilon}\right)^{M-n} +\left(1-e^{-\epsilon}\right)^M   \\ \nonumber &\approx&  \sum^{M}_{n=1} \frac{1}{2^{n-1}(n-1)!}\sum^{n}_{i=0}{n \choose i} (-1)^i \left( 2^{n-1}\sum^{n-1}_{l=0}\frac{(-1)^l(n-l-1)!}{l!}\left(\frac{ib}{P}\right)^l\right. \\ \label{approxi worst channel} &&+\left. \left(\frac{ib}{P}\right)^n 2^{2n} q(\ln x) \right)  \frac{M!}{n!(M-n)!} (1-n\epsilon) \epsilon^{M-n} +\epsilon^M.
\end{eqnarray}
\end{figure*}
Recall that the sum of binomial coefficients has the following properties \cite{GRADSHTEYN}:
  \begin{eqnarray}
  \sum^{n}_{i=0}{n \choose i} (-1)^{i}   i^l &=& 0,
\end{eqnarray}
for $0\leq l < n$, and
  \begin{eqnarray}
  \sum^{n}_{i=0}{n \choose i} (-1)^{i}   i^n &=& (-1)^nn!.
\end{eqnarray}
By using such properties, the expression shown in \eqref{approxi worst channel} can be simplified significantly. Specifically  all the factors of the order of $\left(\frac{ib}{P}\right)^{l} $, $l<n$, will be completely removed,  and we can write
  \begin{eqnarray}\nonumber
\mathcal{P}_{best,II}  &\approx&  \sum^{M}_{n=1} \frac{1}{2^{n-1}(n-1)!}\sum^{n}_{i=0}{n \choose i} (-1)^i  \left( \left(\frac{ib}{P}\right)^n \right. \\ && \nonumber \times \left.2^{2n} q(\ln x) \right)  \frac{M!}{n!(M-n)!} (1-n\epsilon) \epsilon^{M-n} +\epsilon^M\\ \nonumber &\approx&  -c \epsilon^M\ln \epsilon +\epsilon^M,
\end{eqnarray}
and the proposition is proved.
\end{proof}
An interesting observation here is that the high SNR approximation of the outage probability achieved by the equal power allocation scheme also includes a term $\ln \epsilon$, similar to the individual transmission scheme. Compared to traditional relaying networks, such a phenomenon is unique, which is due to the fact that, in energy harvesting cases,  the relaying transmission power is constrained by the source-relay channel attenuation during the first phase transmissions.

\subsection{Worst outage performance}\label{subsection worst}
The worst  outage performance achieved by the non-cooperative individual strategy will be
\begin{eqnarray}
\mathcal{P}_{worst, I}
&\approx&  1 -  (1-M\epsilon) \left( 1+  \frac{2\epsilon}{\eta} \ln  \sqrt{\frac{\epsilon}{\eta  }} \right)^M \\ \nonumber &\approx& M\epsilon \left(1- \frac{2}{\eta} \ln  \sqrt{\frac{\epsilon}{\eta  }}\right),
\end{eqnarray}
which   still decays  as $\frac{\log SNR}{SNR}$.
 And the worst outage performance achieved by the equal power allocation can be approximated as
    \begin{eqnarray}
\mathcal{P}_{worst,II} &\approx&\nonumber
\frac{M}{M-1} \frac{Mb}{P} e^{-M\epsilon} +1-e^{-M\epsilon}\\ \nonumber & \approx& \epsilon M\left(1+\frac{M}{\eta(M-1)}\right),
\end{eqnarray}
which   decays as $\frac{1}{SNR}$, and hence realizes less error than the individual transmission scheme.   According to Lemma \ref{lemma max min}, the water filling strategy can achieve the optimal performance for the user with the worst channel conditions. Upper and lower bounds have been developed in Lemma \ref{lemma worst}.   In the following we   show  that  such bounds converge at high SNR. We first focus on the upper bound which can be rewritten as
\begin{align}\nonumber
\mathcal{P}_{worst, III}&<1- e^{-M\epsilon}  \left(     \left(1- \frac{1}{(M-1)}  \frac{\epsilon M^2}{\eta}\right) \right. \\\nonumber   & +\left.M \left(\frac{\epsilon}{\eta}\right)\int^{M-1}_{c}  \left(\frac{1}{M-1}\right.\right. \\ \nonumber &  \left.\left. - \frac{a(y)\epsilon}{\eta} \frac{1}{(M-1)(M-2)}\right)dy\right)\\\nonumber
&\approx 1- e^{-M\epsilon}  \left(     \left(1- \frac{1}{(M-1)}  \frac{\epsilon M^2}{\eta}\right) \right. \\\nonumber   & +\left.M \left(\frac{\epsilon}{\eta}\right)\int^{M-1}_{c}   \frac{1}{M-1}  dy\right)
\\\nonumber
&\approx \epsilon \left(M+\frac{M(M-1-c)}{(M-1)\eta}+\frac{M^2}{(M-1)\eta}\right),
\end{align}
for $M>2$, where the first inequality follows from \eqref{approximation} and $c$ is a constant to ensure $a(y)\epsilon\rightarrow 0$ for $y\in[c, M-1]$.
As can be observed from the above equation, the upper bound  decays as $\frac{1}{SNR}$. The same conclusion can be obtained for the case of $M=2$ by applying the approximation in \eqref{approximation2}.  On the other hand, the lower bound of the worst case probability can be similarly approximated as follows:
\begin{eqnarray}\nonumber
\mathcal{P}_{worst, III}> \epsilon M\left(1+\frac{1}{\eta(M-1)} \right).
\end{eqnarray}
Since both upper and lower bounds decay   as $\frac{1}{SNR}$, we can conclude that the worst case outage probability achieved by the water filling scheme   also decays  as $\frac{1}{SNR}$, which also shows the   rate of decay of the averaged outage probability for the water filling scheme.

\section{Auction based Power Allocation}\label{section auction}
In the previous sections, three different strategies    to use  the harvested energy have been studied, where the water filling strategy can achieve the best performance in several criteria.  However, such a centralized method requires that the relay has   access to   global CSI. For a system with a large number of users,  the provision of global CSI  consumes significant system overhead, which motivates the study of the following auction based    strategy to realize distributed power allocation.

\subsection{Power auction game}
The addressed power allocation problem can be modeled as a game in which  the multiple destinations  compete with each other for the assistance of the relay. Note that we will  need to  consider only the destinations whose corresponding source messages can be reliably decoded at the relay.  Specifically each destination   submits a bid  to the relay, and the relay will update the power allocation of the users  at the end of each iteration. Each destination  knows only its own channel information, and the relay has no access to   relay-destination channel information. The described game can be formulated as   follows:
\begin{itemize}
\item Bids: Each user submits a scalar $b_i$ to the relay;
\item Allocation: The relay will allocate the following transmission power to each user:
\begin{equation}
P_{ri} = \frac{b_i}{\sum^{M}_{j=1}b_j+\xi} P_r = \frac{b_i}{\sum^{M}_{j=1}b_j+\xi} \sum^{M}_{i=1}  \eta P_s |h_{i}|^2\theta_i,
\end{equation}
where $\xi$ is a factor related to the power reserved at the relay.
\item Payments: Upon the allocated transmission power $P_{ri}$,   user $i$ pays the relay $C_i=\pi P_{ri}$.
\end{itemize}
Recall that the data rate the $i$-th  destination can achieve is $
R_{d,i}= \frac{1}{2}\log(1+ {P_{ri}} |g_i|^2)$, and therefore it is natural  to consider a game in which  the $i$-th user selects $b_i$ to maximize its payoff as follows:
\begin{eqnarray}
U_i(b_i; \mathbf{b}_{-i},\pi)   &=&
\frac{1}{2}\log(1+  P_{ri} |g_i|^2) - \pi P_{ri},
\end{eqnarray}
where  $\mathbf{b}_{-i}=(b_1, \cdots, b_{i-1}, b_{i+1}, b_N)$. The addressed game in strategic form is a triplet $G=\left(\mathcal{N}, (b_i)_{i\in \mathcal{N}}, (U_i)_{i\in \mathcal{N}} \right)$, where $\mathcal{N}=\{1, \ldots, N\}$ includes all the destinations whose  source messages can be delivered to the relay successfully, and $b_i\in \mathbb{R}^+$, where  $\mathbb{R}^+$ denotes the nonnegative real numbers. A desirable outcome of such a game is the following:
\begin{Definition}
The Nash Equilibrium (NE) of the addressed game, $G$, is a bidding profile $\mathbf{b}^*$ which ensures that no user wants to deviate unilaterally, i.e.,
\begin{equation}
U_i(b_i^*; \mathbf{b}^*_{-i}, \pi) \geq  U_i(b_i; \mathbf{b}^*_{-i}, \pi), \forall i \in \mathcal{N}, \forall b_i \geq 0.
\end{equation}
\end{Definition}
The following proposition provide the best response functions and the uniqueness of the   NE for the addressed game.

\begin{proposition}\label{propoistion}
For the addressed power auction game,   there exists a threshold price such that a unique NE exists when the price $\pi$ is larger than such a threshold, otherwise there are infinitely many equilibria. In addition, the unique best response function for each player $i$ can be expressed as in \eqref{best response function},
\begin{figure*}
\begin{eqnarray}\label{best response function}
BR_i(\mathbf{b}_{-i},\pi) =\left\{ \begin{array}{ll} \infty, &  \pi< \min \{\pi_{u,i}, i\in \mathcal{N}\} \\
0, & \pi> \max\left\{ \frac{|g_i|^2}{2\ln 2},  i\in \mathcal{N}\right\}  \\\frac{\left( \frac{1}{2\ln 2 \pi} - \frac{1}{|g_i|^2}\right)}{P_r - \frac{1}{2\ln 2 \pi} - \frac{1}{|g_i|^2} } \left(\sum^{M}_{j=1,j\neq i}b_j +\xi\right),& otherwise \\
\end{array} \right.,
\end{eqnarray}
\end{figure*}
where $\pi_{u}=\frac{|g_i|^2}{2\ln 2(1+P_{r}|g_i|^2)} $.
\end{proposition}
\begin{proof}
See the appendix.
\end{proof}
Some explanations  about the choice of the best response function shown in Proposition \ref{propoistion} follows. When the price is too large, a player's payoff function is always negative, and therefore it simply quits the game, i.e.  $b_i=0$. Another extreme case is that the price is too small, which motivates a player to compete aggressively with other players by using a large bid. Unlike \cite{Huangjianwei06} and \cite{Huangj08}, the uniqueness of the NE is shown by using the contraction mapping property of the best response functions, which can simplify the following discussion about practical  implementation.

The addressed power auction game can be implemented in an iterative way. The relay will first announce the price to all players. During each iteration, each user will update its bid according to the following:
\begin{eqnarray}\label{iter}
BR_i(\mathbf{b}_{-i}(n-1))  &=&
\frac{\left( \frac{1}{2\ln 2 \pi} - \frac{1}{|g_i|^2}\right)}{P_r - \frac{1}{2\ln 2 \pi} - \frac{1}{|g_i|^2} }\\ \nonumber &&\times \left(\sum^{M}_{j=1,j\neq i}b_j^{n-1} +\xi\right),
\end{eqnarray}
where $BR_i(\mathbf{b}_{-i}(n-1)) $ quantifies  the best response dynamics determined by the  actions from the previous iteration. For simplicity,  we  consider only the case of  $  \max\left\{ \frac{|g_i|^2}{2\ln 2},i\in \mathcal{N}\right\}\leq \pi\leq \min \{\pi_{u,i}, i\in \mathcal{N}\}$. As shown in the proof for Proposition \ref{propoistion}, the best response function for the addressed power auction game is a contraction mapping, provided that the price is larger than the threshold, which means that this iterative algorithm   converges to a unique fixed point, namely  the NE of the addressed game. Note that such a convergence property is proved without the need for the nonnegative matrix theory as in \cite{Huangjianwei06} and \cite{Huangj08}.

In practice, the implementation of the iterative steps in \eqref{iter} requires a challenging assumption that each user knows the other players' actions, and such an assumption can be avoided by using the following equivalent updating function,
$
BR_i(\mathbf{b}_{-i}(n-1)) =
\frac{\left( \frac{1}{2\ln 2 \pi} - \frac{1}{|g_i|^2}\right)}{P_r - \frac{1}{2\ln 2 \pi} - \frac{1}{|g_i|^2} }\frac{(P_r-P_{ri}^{n-1}) b_j^{n-1}}{P_{ri}^{n-1}}$, where we only consider the case  $ \left( \frac{|g_i|^2}{2\ln 2(1+P_{ri}|g_i|^2)} -\pi _i\right)>0 $.
As a result,  each user can update its bid only according to its local information, such as its previous allocated power and previous bid, without needing to know the actions of other users.
\section{Numerical results}\label{section numberical results}
In this section, computer simulations will be carried out to evaluate the performance of those energy harvesting relaying protocols described in the previous sections.

 We first study the accuracy of the developed analytical results. Specifically in Fig. \ref{fig compare},  the outage performance achieved by the individual transmission scheme and the equal power allocation scheme is shown as a function of SNR. All the channel coefficients are assumed to be complex Gaussian with zero means and unit variances. The targeted data rate is $R=2$ bits per channel use (BPCU), and the energy harvesting efficiency is set as   $\eta=1$. As can be seen from the figures   the  developed analytical results   exactly match   the simulation results, which demonstrates the accuracy of the developed analytical results.

\begin{figure}[!htp]
\begin{center} \subfigure[ Individual Transmission]{\includegraphics[width=0.48\textwidth]{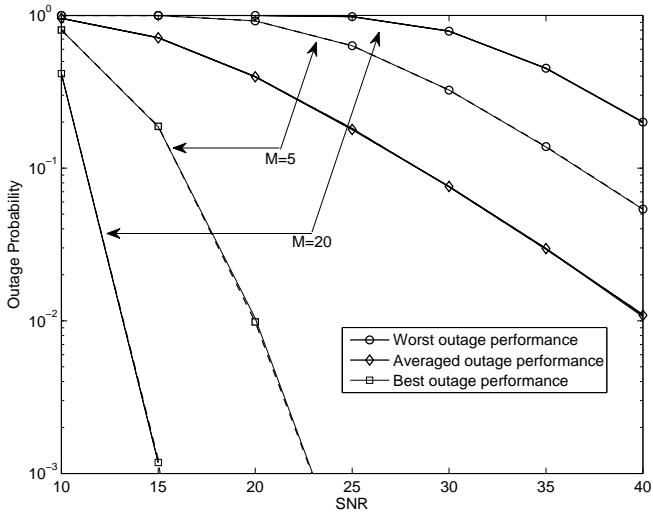}}
\subfigure[Equal Power Allocation  ]{\label{fig set comparison
b2}\includegraphics[width=0.48\textwidth]{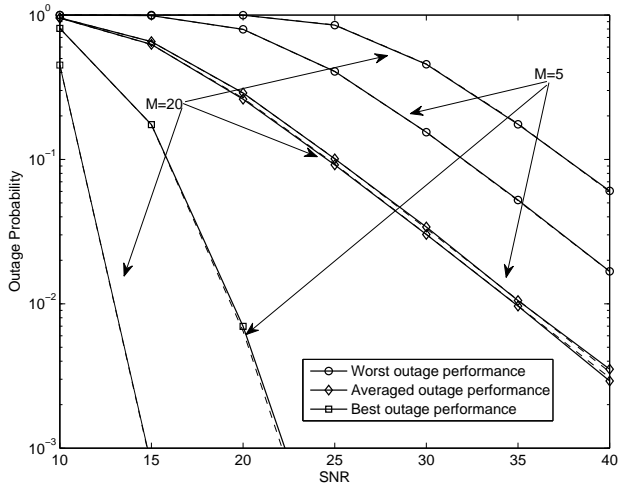}}
\end{center}
 \caption{ Outage probabilities achieved by the individual transmission scheme and the equal power allocation scheme. $R=2$ BPCU. The solid curves are for the simulation results and the dashed ones are for the analytical results. }\label{fig compare}
\end{figure}

Comparing the two cases  in Fig. \ref{fig compare}, we   find that the use of the equal power allocation strategy   improves the outage performance. Consider  the averaged outage performance as an example. When the SNR is $40$ dB, the use of the individual transmission scheme  realizes   outage probability of $1\times 10^{-2}$, whereas the equal power allocation scheme can reduce  the  outage probability to $3\times 10^{-3}$. Such a phenomenon confirms the asymptotic results shown in Section IV.A. Specifically the outage probability achieved by the individual transmission scheme decays with the SNR at a rate $\frac{1}{SNR}(1+2\ln SNR)$, but the equal power allocation scheme can achieve a faster   rate of decay, $\frac{1}{SNR}$.

When more source-destination pairs join in the transmission, it is more likely to have some nodes with extreme channel conditions, which is the reason to observe the phenomenon in Fig. \ref{fig compare} that with   a larger number of user pairs, the best outage performance improves but the worst outage performance degrades.  The impact of the number of user pairs on the average outage performance can also be observed in the figure. For the individual transmission scheme, there is no cooperation among users, so the number of user pairs has no impact on the average outage performance. On the other hand, it is surprising to find that an increase in the   number of users  yields  only a slight improvement in  the performance  of the equal power allocation scheme, which might be due to the fact that power allocation can improve the transmission from the relay to the destinations, but not   the source-relay transmissions.

\begin{figure}[!htp]
\begin{center} \subfigure[ Bounds for worst outage performance]{\includegraphics[width=0.48\textwidth]{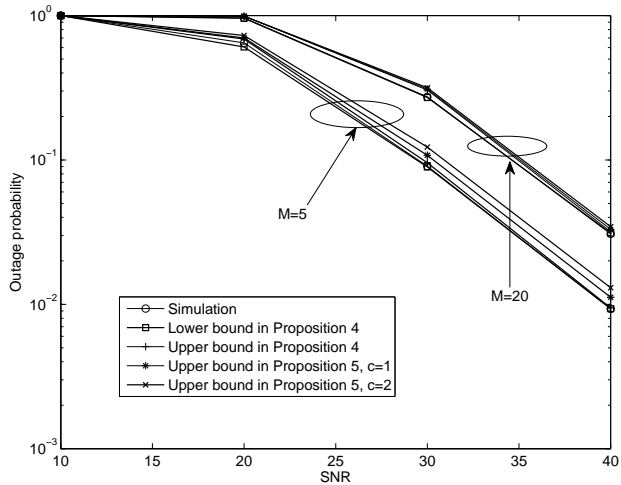}}
\subfigure[Outage performance   ]{\label{fig set comparison
b2}\includegraphics[width=0.48\textwidth]{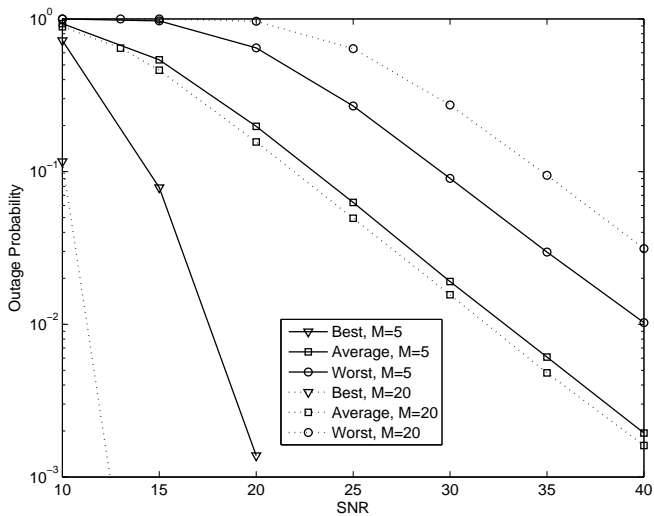}}
\end{center}\vspace{-1em}
 \caption{Outage probability achieved by the water filling  transmission scheme. $R=2$ BPCU.  }\label{fig water}\vspace{-2em}
\end{figure}

In Fig. \ref{fig water}, the performance of the water filling scheme is studied. The same simulation setup as in the previous figures is used.  Firstly the  upper and lower bounds developed in Propositions  \ref{lemma worst} and \ref{lemma worst2} are compared to the simulation results in Fig. \ref{fig water}.a. As can be seen from the figure, the lower bound developed in \eqref{lowerbound1} and the upper bound in \eqref{upperbound 1} are very tight. Recall that the reason for the bounds in \eqref{upperbound 1} and \eqref{lowerbound1} are tight is because the dominant factor in the summation $\sum^{M}_{i=1}z_i$ is $z_{(M)}$. As shown at the end of Section \ref{section centralized},   $z_{(M)}$ is  an unbounded variable, whereas the variance of the   other variables,  i.e. $z_{(m)}$, $m<M$, are always bounded.   In Fig. \ref{fig water}.b, the outage performance based on different criteria is  shown for the water filling scheme. Comparing Figures  \ref{fig compare} and \ref{fig water}, we can see that the use of the water filling scheme   yields the best performance.

\begin{figure}[!htbp]\centering
    \epsfig{file=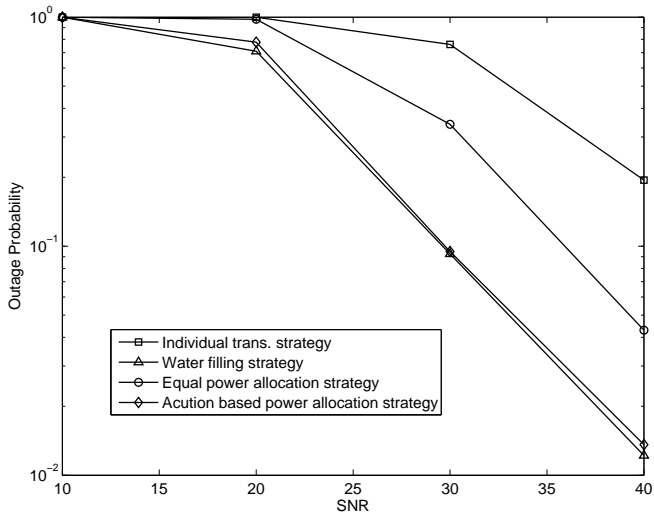, width=0.48\textwidth, clip=}
\caption{Outage probability for the user with the worst channel conditions. $R=0.5$ BPCU.}\label{fig_worst}\vspace{-1em}
\end{figure}

In Figures \ref{fig_worst}, \ref{fig_averged} and \ref{fig number}, we focus on the comparison among the different power allocation strategies described in this paper. The targeted data rate is $0.5$ BPCU, and there are $20$ user pairs, i.e.,  $M=20$. Channels are  assumed to be Rayleigh fading with path loss attenuation. Particularly it is assumed that the distance from the sources to the relay is $2$m,  the same as the distance from the relay to the destinations.   In Fig. \ref{fig_worst}, we study the outage performance for the user with the worst channel conditions. The water filling scheme can outperform the individual transmission and equal power strategies, consistent with the observations from the previous figures. The auction based strategy can achieve  performance   close to that of the water filling scheme. As indicated in Lemma \ref{lemma max min}, the water filling scheme is optimal in terms of the outage performance for the user with the worst channel conditions, which implies that  the auction based scheme can achieve  performance close to the optimal.
\begin{figure}[!htbp]\centering
    \epsfig{file=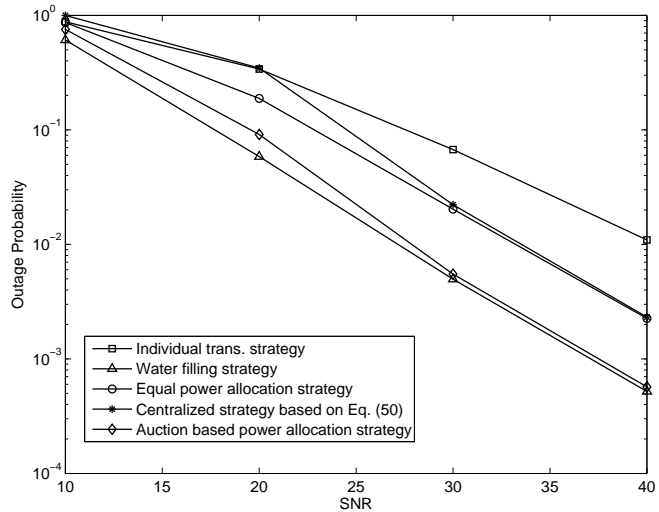, width=0.48\textwidth, clip=}\vspace{-1em}
\caption{Average outage performance achieved by the studied transmission protocols.   $R=0.5$ BPCU.  }\label{fig_averged}\vspace{-1em}
\end{figure}

In Fig. \ref{fig_averged} the averaged outage performance achieved by the addressed transmission schemes is shown. Similar to Fig. \ref{fig_worst}, the water filling scheme can achieve the best performance, and the auction based scheme has  performance close to the water filling scheme. In Fig. \ref{fig_averged}, we also show the performance of the centralized strategy based on \eqref{maximi worst}. Such a scheme can achieve the optimal performance for the user with the worst channel conditions, which means that  its performance will be the same as the water filling scheme shown in Fig. \ref{fig_worst}. However, such a scheme suffers some performance loss when the criterion is changed to the averaged outage performance. Finally, the performance of the transmission schemes are compared in Fig. \ref{fig number} in terms of the number of successful receivers. In this context, the auction based scheme can achieve   better performance than the equal power one, by realizing two more successful receivers for the SNR range from $10$dB to $25$dB. An interesting observation is that the strategy maximizing the worst outage performance realizes fewer  successful receivers than the individual scheme, which is due to the fact that such a strategy will put the user with the worst channel conditions as the top priority. And allocating more power to such users with poor channel conditions will reduce the performance of   other users.   Consistently with  the previous figures,  the water filling scheme can achieve the best performance, and ensure  the most successful receivers. However, it is worth   recalling  that the water filling scheme requires   global CSI at the relay, whereas the other schemes, such as the auction based and equal power strategies, can be realized in a distributed way.

\begin{figure}[!htbp]\centering
    \epsfig{file=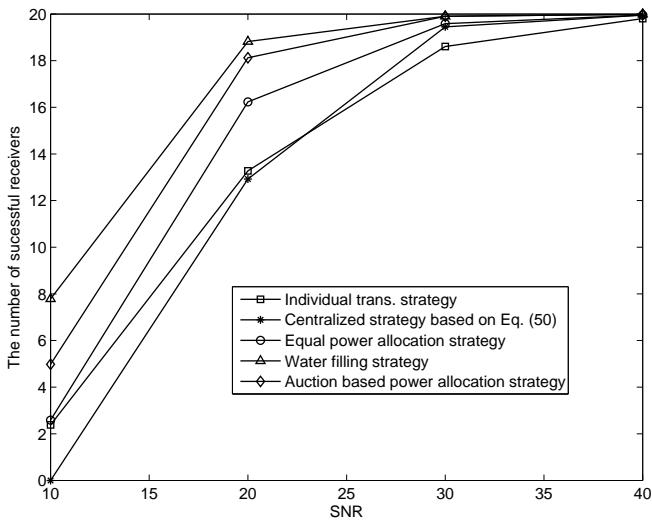, width=0.48\textwidth, clip=}\vspace{-1em}
\caption{ The number of successful destinations as a function of SNR   $R=0.5$ BPCU}\label{fig number}\vspace{-2em}
\end{figure}
\section{Conclusion} \label{section conclusion}
 In this paper, we have considered several power allocation strategies for a cooperative network in which multiple source-destination pairs communicate with each other via an energy harvesting relay.  The non-cooperative individual transmission scheme results in a  outage performance decaying as $\frac{\log SNR}{SNR}$,  the centralized power allocation strategies   ensure that the outage probability decays at a faster rate $\frac{1}{SNR}$, and the water filling scheme can achieve   optimal performance in terms of a few criteria.   An auction based power allocation scheme has  also been   proposed to achieve a better tradeoff between the system performance and complexity.

\section{Acknowledgments }
The authors thank Dr Zhiyong Chen for helpful discussions.

 \appendix
 \textsl{Proof of Theorem \ref{theorem1}:}
According to the instantaneous realization of the channels, we can group destinations into two sets, denoted by $\mathcal{S}_1$ and $\mathcal{S}_2$. $\mathcal{S}_1$ includes the destinations whose corresponding sources cannot deliver their information reliably to the relay, and   $\mathcal{S}_2$ includes the remaining  destinations; thus the size of $\mathcal{S}_2$ is $N$, i.e. $|\mathcal{S}_2|=N$. Therefore the outage probability for the $i$-th   destination is
  \begin{eqnarray}\nonumber
\mathcal{P}_{i,II} =\sum^{M}_{n=1} \mathrm{Pr}\left( \frac{1}{2}\log\left(1+\frac{P_r}{n}|g_{i}|^2\right)<R, N=n, i\in\mathcal{S}_2 \right)\\\nonumber+\mathrm{Pr}(i\in\mathcal{S}_1 ).
\end{eqnarray}
The second probability on  the righthand side of the above equation can be  calculated as $(1-e^{-\epsilon})$ by analyzing the error event $|h_i|^2<\epsilon$. The probability of the event $i\in \mathcal{S}_2$ is $\frac{n}{M}$, conditioned on the size of the subset   $N$, so the first factor in the above equation can be rewritten as
  \begin{eqnarray}\label{eq 22}
\mathcal{P}_{i,II} =
\sum^{M}_{n=1}\frac{n}{M} \mathrm{Pr}\left( \left.\frac{1}{2}\log\left(1+\frac{P_r}{N}|g_{i}|^2\right)<R\right|N=n \right)\\ \nonumber \times\mathrm{Pr}(N=n).
\end{eqnarray}
The total available energy given $N$, the size of $\mathcal{S}_2$,  is
$$P_r = \sum^{n}_{i=1}\eta P_s|h_{i}|^2\theta_i =\sum^{n}_{i=1}\eta \left(P_s|h_{i}|^2-   2^{2R}+1 \right).$$ Define $Q_1\triangleq  \mathrm{Pr}\left( \left.\frac{1}{2}\log\left(1+\frac{P_r}{N}|g_{i}|^2\right)<R\right|N=n \right)$ which can be written as
 \begin{eqnarray}
 Q_1 =  \mathrm{Pr}\left(  \left.\left( P_s\sum^{n}_{i=1}| {h}_{i}|^2-na\right) |g_{i}|^2 <\frac{na}{\eta} \right| |{h}_{i}|^2>\epsilon,\right. \\ \nonumber \left. 1\leq i\leq n, |{h}_{i}|^2<\epsilon, n+1\leq i\leq M \right).
\end{eqnarray}
 Using   the independence among the channels, $Q_1$ can be evaluated as
  \begin{eqnarray}
 Q_1 =  \mathrm{Pr}\left(  \left.\left( P_s\sum^{n}_{i=1}| {h}_{n_i}|^2-na\right) |g_{i}|^2 <\frac{na}{\eta} \right| |{h}_{n_i}|^2>\epsilon,\right.\\ \nonumber \left. 1\leq i\leq n \right).
\end{eqnarray}
Define $Y=\sum^{n}_{i=1}| {h}_{n_i}|^2$. To evaluate the above probability, it is important to find the density function of the sum of $n$ exponentially distributed variables, $Y$, with the condition  that each variable is larger than $\epsilon$.   Conditioned on $|{h}_{n_i}|^2>\epsilon$, we can find the Laplace transform of the density function of $|{h}_{n_i}|^2$ as
\begin{eqnarray}\nonumber
\mathcal{L}_{|{h}_{n_i}|^2>\epsilon}(s) &=& \frac{1}{1-F_{|h_i|^2}(\epsilon)}\int^{\infty}_{\epsilon} e^{-xs} f_{|h_i|^2}(x)dx\\ \nonumber &=& e^{\epsilon} \int^{\infty}_{\epsilon} e^{-xs} e^{-x}dx = \frac{1}{1+s}e^{-s\epsilon}.
\end{eqnarray}
Given the independence among the channels, conditioned on $ |{h}_{n_i}|^2>\epsilon, 1\leq i\leq n$, the density function of the sum of these channel coefficients has the following Laplace transform:
\begin{eqnarray}
\mathcal{L}_{\sum^{N}_{i=1}| {h}_{n_i}|^2}(s) = \left(\mathcal{L}_{|{h}_{n_i}|^2>\epsilon}(s) \right)^n = \frac{e^{-n\epsilon s}}{(1+s)^n}.
\end{eqnarray}
By inverting  Laplace transform, the pdf  of the sum, conditioned on $ |{h}_{n_i}|^2>\epsilon, 1\leq i\leq n$, is obtained as
\begin{eqnarray}\label{pdf}
f_{\sum^{N}_{i=1}| {h}_{n_i}|^2}(y) = \frac{(y-n\epsilon)^{n-1}}{(n-1)!}  e^{-(y-n\epsilon)}u(y-n\epsilon).
\end{eqnarray}
A special case   is when $\epsilon=0$, in which case  the above expression reduces to the   classical chi-square distribution.
Now the addressed probability can be calculated as
 \begin{eqnarray}\nonumber
 Q_1 &=&  \int_{n\epsilon } ^{\infty} \left(1- e^{-\frac{b_n}{Py-na}}\right) \frac{1}{(n-1)!} (y-n\epsilon)^{n-1} e^{-(y-\epsilon)}dy
 \\   \nonumber
 &=& \frac{1}{(n-1)!}\left((n-1)! -  2\left(\frac{b_n}{P}\right)^{\frac{n}{2}}\mathbf{K}_{n}\left(2\sqrt{\frac{b_n}{P}}\right) \right).
\end{eqnarray}
So the overall outage probability can be obtained after  some algebraic manipulations by using  the following result:
   \begin{eqnarray}\label{eq32}
\mathrm{Pr}(N=n)  &=& \frac{M!}{n!(M-n)!} e^{-n\epsilon}\left(1-e^{-\epsilon}\right)^{M-n}.
\end{eqnarray}
And the theorem is proved.
 \hspace{\fill}$\blacksquare$\newline

\textsl{Proof of Lemma \ref{lemma max min} :}
The lemma can be proved by first developing a power allocation strategy optimal to the worst user outage performance and then showing that such a scheme achieves the same worst user outage probability  as the water filling strategy.

Suppose that  there are  $n$ sources that  can deliver their signals   successfully  to the relay.  The power allocation  problem, which is to optimize the worst user outage performance,  can be formulated as follows:
\begin{eqnarray}\label{maximi worst}
\underset{P_{ri}}{\max} & \min \{R_{d,1}, \cdots, R_{d,n}\}\\ \nonumber s.t. & \sum_{i=1}^nP_{ri} = P_r.
\end{eqnarray}
In order to find a closed-form expression for its solution, this  optimization problem can be converted into the following equivalent form by introducing an auxiliary parameter:
\begin{eqnarray}
\underset{P_{ri}}{\max}  & t\\ \nonumber s.t. &  \frac{1}{2} \log(1+P_{ri}|g_{i}|^2)>t \\ \nonumber  & \sum_{i=1}^nP_{ri} = P_r.
\end{eqnarray}
By applying the Karush-Kuhn-Tucker conditions \cite{Boyd}, a closed form expression for  the optimal  solution can be obtained as
\begin{eqnarray}
P_{ri} = \frac{2^{2t}-1}{|g_{i}|^2}.
\end{eqnarray}
And the parameter $t$ can be found by solving the following equation based on  the total power constraint:
\begin{eqnarray}
\sum^{n}_{i=1}\frac{2^{2t}-1}{|g_{i}|^2} = \sum^{n}_{i=1}  \eta P_s |h_{i}|^2\theta_i,
\end{eqnarray}
which yields
\begin{eqnarray}
\label{eq 5}
t&=&\frac{1}{2} \log\left(1+\frac{\sum^{n}_{i=1}  \eta P_s |h_{i}|^2\theta_i}{\sum^{n}_{i=1}\frac{1}{|g_{i}|^2} }\right)\\\nonumber
&=&\frac{1}{2} \log\left(1+\frac{\sum^{n}_{i=1}  \eta (P_s |h_{i}|^2-   2^{2R}+1)}{\sum^{n}_{i=1}\frac{1}{|g_{i}|^2}} \right).
\end{eqnarray}
By using this closed form solution,  the worst user outage probability  can be written as in \eqref{eq7}.
\begin{figure*}
\begin{eqnarray}
\label{eq7}
\mathcal{P}_{out} & =&    \mathrm{Pr}\left(\min\left\{\frac{1}{2}\log(1+P_s|h_{i}|^2), i\in\{1, \cdots, M\}\right\}>R  ,N=M\right) + \sum^{M-1}_{n=0}\mathrm{Pr}(N=n)
\\ \nonumber &=&    \mathrm{Pr}\left( \frac{\sum^{M}_{i=1}  \eta (P_s |h_{i}|^2-   2^{2R}+1)}{\sum^{M}_{i=1}\frac{1}{|g_{i}|^2}} <2^{2R}-1, N=M \right) + \sum^{M-1}_{n=0}\mathrm{Pr}(N=n).
\end{eqnarray}
\end{figure*}
On the other hand, for the addressed water filling strategy, the outage event for the user with the worst performance rises either because  at least one of the source messages cannot be detected at the relay, $N<M$, or  there is not enough power for all users, which means that the outage probability will be
\begin{eqnarray}
\label{eq8}
\mathcal{P}_{worst,III} =   \mathrm{Pr}\left( \sum^{M}_{i=1}\eta \left(P_s|h_{i}|^2-  a \right) <  \sum^{M}_{i=1} \frac{2^{2R}-1}{|g_i|^2}\right., \\ \nonumber \left.  N=M\right) + \sum^{M-1}_{n=0}\mathrm{Pr}(N=n).
\end{eqnarray}
Comparing \eqref{eq7} and \eqref{eq8}, we find that two strategies achieve the same worst outage performance, and the lemma is proved.  \hspace{\fill}$\blacksquare$\newline

\textsl{Proof of  Proposition \ref{lemma worst} :}
The expression for   the outage probability of  the user with the worst channel conditions achieved by the water filling strategy is given  in \eqref{eq8}. The first factor in the expression, denoted by  $Q_4$,  can be expressed as
\begin{eqnarray}
\label{eq02}
Q_4 =  \mathrm{Pr}\left( \left. \sum^{M}_{i=1} z_i >\frac{1}{a}\sum^{M}_{i=1}\eta \left(P_s|h_{i}|^2-  a \right)\right| N=M\right) \\ \nonumber \times \mathrm{Pr}(N=M).
\end{eqnarray}
To obtain some insightful understandings for the water filling  scheme, we consider the following bounds:
\begin{eqnarray}\label{bounds used in worst}
\mathcal{P}_w\left(   z_{(M)} >w \right) < \mathcal{P}_w\left(  \sum^{M}_{i=1} z_i >w \right) \leq \mathcal{P}_w\left(   z_{(M)} \right. \\ \nonumber \left.+(M-1)z_{(M-1)}>w \right),
\end{eqnarray}
where $w=\frac{1}{a}\sum^{M}_{i=1}\eta \left(P_s|h_{i}|^2-  a \right)$, $\mathcal{P}_w(\cdot)$ denotes the probability conditioned on a fixed $w$, and the condition $N=M$ has been omitted to simplify  notation.
 The upper bound can be written as
\begin{eqnarray}\nonumber
\mathcal{P}_w\left(   z_{(M)} +(M-1)z_{(M)}>w \right) = \mathcal{P}_w\left(   z_{(M)}  >w \right) \\ \nonumber +\mathcal{P}_w\left(   z_{(M)} +(M-1)z_{(M-1)}>w, \frac{w}{M}<z_{(M)} <w \right),
\end{eqnarray}
where the condition $z_{(M)}>\frac{w}{M}$ is due to the fact that $z_{(M)}$ is the largest among the $M$ ordered variables.
Denote the second probability on the righthand side of the above equation conditioned on a fixed $w$ by $Q_3$. Recall that the joint probability density function (pdf) of two ordered statistics $z_{(i)}$ and $z_{(j)}$, $1\leq i<j\leq M$, can be written as \cite{David03}
\begin{eqnarray}
f_{z_{(i)},z_{(j)}}(u,v) = \frac{M!}{(i-1)!(j-1-i)!(M-j)!} f(u)f(v)\\ \nonumber\times (F(u))^{i-1}\left(F(v) - F(u)\right)^{j-1-i}(1-F(v))^{M-j},
\end{eqnarray}
 where the pdf and cumulative distribution function (CDF) are defined in \eqref{exxx} and the subscript $z_i$ has been omitted for simplicity. Based on such a pdf, the  probability $Q_3$  can be written as
\begin{eqnarray}\nonumber
Q_3 &=& M(M-1)\int^{w}_{\frac{w}{M}} \int^{v}_{\frac{w-v}{M-1}} f(u)F(u)^{M-2} du f(v)dv \\ \nonumber &=& M\int^{w}_{\frac{w}{M}} \left(F(v)^{M-1}- F\left(\frac{w-v}{M-1}\right)^{M-1}\right) f(v)dv.
\end{eqnarray}
Substituting the density function of $z_i$, we obtain
\begin{eqnarray}\label{Q3}
Q_3 = \left(e^{-\frac{M}{w}}-e^{-\frac{M^2}{w}}\right) -M\int^{w}_{\frac{w}{M}}   e^{-\frac{(M-1)^2}{w-v}} \frac{1}{v^2} e^{-\frac{1}{v}}dv.
\end{eqnarray}
 The probability $\mathcal{P}_w\left(   z_{(M)}  >w \right)$ can be obtained by applying the pdf of the largest order statistics as
$
\mathcal{P}_w\left(   z_{(M)}  >w \right) = 1 - e^{-\frac{M}{w}}$.
So conditioned on a fixed $w$, the upper bound can be expressed as
\[
\mathcal{P}_w\left(   z_{(M)} +(M-1)z_{(M)}>w \right)\leq 1 - e^{-\frac{M}{w}}+Q_3.
\]
On the other hand, conditioned on  $M$ source messages successfully decoded at the relay, the density function of $w=(\eta\frac{1}{\epsilon}\sum^{M}_{i=1}   |h_{i}|^2- M\eta  ) $ can be obtained from \eqref{pdf} as $
f_w(w) = \frac{1}{(M-1)!}\left(\frac{\epsilon}{\eta}\right)^M w^{M-1} e^{-\frac{\epsilon}{\eta}w}$. So the upper bound can be expressed as
\begin{eqnarray}\nonumber
\int^{\infty}_{0}\left(1 - e^{-\frac{M}{w}}+Q_3\right)f_w(w)dw \cdot \mathrm{Pr}(N=M)\\ \nonumber+\sum^{M-1}_{n=0}\mathrm{Pr}(N=n),
\end{eqnarray}
and the first part of the proposition is proved. The lower bound can be proved by using the   steps similar to those used in the proof of Proposition \ref{lemma worst2}, and will be omitted here.  \hspace{\fill}$\blacksquare$\newline

\textsl{Proof of Proposition \ref{lemma worst2} :}
Recall that the upper bound for the water filling scheme is
\begin{eqnarray}\nonumber
\mathcal{P}_{worst, III}<e^{-M\epsilon} \int_{0}^{\infty}\left(1-       e^{-\frac{M^2}{w}}   \right. \\   -M\left.\underset{Q_5}{\underbrace{\int^{w}_{\frac{w}{M}}   \frac{ e^{-\frac{(M-1)^2}{w-v}-\frac{1}{v}}}{v^2} dv}}\right)f_w(w)dw  +1-e^{-M\epsilon}.
\end{eqnarray}
To obtain a more explicit expression for this upper bound, the factor $Q_5$ can be rewritten as
\begin{align} \nonumber
Q_5  &= \frac{1}{w}\int^{M-1}_{0} e^{-\frac{(y+1)(M-1)^2}{wy} - \frac{y+1}{w}}dy \\ \label{new equation} &=  \frac{1}{w}\int^{M-1}_{0}   e^{-\frac{a(y)}{w}}dy,
\end{align}
where $a(y)= (y+1)\left(\frac{(M-1)^2+1}{y}\right)$. An important observation from \eqref{new equation} is that $a(y)$ is not a function of $w$. Furthermore,  the integration range in \eqref{new equation} is also not a function of $w$. As a result, we can first calculate the integral for $w$ by treating $y$ as a constant. First substituting \eqref{new equation} into the probability expression to obtain  the following:
\begin{align}\label{eq10}
\mathcal{P}_{worst, III}&<1-e^{-M\epsilon}+e^{-M\epsilon} \int_{0}^{\infty}\left(1-       e^{-\frac{M^2}{w}}   \right. \\   &-\left.\frac{M}{w}\int^{M-1}_{0}   e^{-\frac{a(y)}{w}}dy\right)f_w(w)dw   .
\end{align}
We focus on the integral of the third  factor in the bracket, denoted by $Q_6$, which is
\begin{eqnarray}\nonumber
Q_6&\triangleq & e^{-M\epsilon} \int_{0}^{\infty}    \frac{M}{w}\int^{M-1}_{0}   e^{-\frac{a(y)}{w}}dy f_w(w)dw \\\nonumber&=&Me^{-M\epsilon}\int^{M-1}_{0} \int_{0}^{\infty}    \frac{1}{w}   e^{-\frac{a(y)}{w}}\frac{1}{(M-1)!}\left(\frac{\epsilon}{\eta}\right)^M \\ \nonumber &&\times w^{M-1} e^{-\frac{\epsilon}{\eta}w}dwdwdy\\\nonumber&=&\frac{M}{(M-1)!}e^{-M\epsilon}\left(\frac{\epsilon}{\eta}\right)^M\int^{M-1}_{0} \int_{0}^{\infty}     w^{M-2}    \\ \nonumber &&\times e^{-\frac{a(y)}{w}-\frac{\epsilon}{\eta}w}dwdwdy
\\\nonumber&=&\frac{M}{(M-1)!}e^{-M\epsilon}\left(\frac{\epsilon}{\eta}\right)\int^{M-1}_{0}    2  \left(\frac{a(y)\epsilon}{\eta}\right)^{\frac{M-1}{2}} \\ \nonumber &&\times \mathbf{K}_{M-1}\left(2\sqrt{\frac{a(y)\epsilon}{\eta}}\right)dy.
\end{eqnarray}
Similarly the integrals of other components in \eqref{eq10} can be evaluated,   the upper bound on the worst outage probability is obtained, and the proposition is proved.
  \hspace{\fill}$\blacksquare$\newline

\textsl{Proof for Proposition \ref{propoistion} :}
The proposition can be   proved by showing the first derivative of the payoff function is
\begin{eqnarray}
\frac{\partial U_i(b_i; \mathbf{b}_{-i},\pi)}{\partial b_i}  =\left( \frac{|g_i|^2}{2\ln 2(1+P_{ri}|g_i|^2)} -\pi _i\right) \frac{\partial P_{ri}}{\partial b_i},
\end{eqnarray}
where $
\frac{\partial P_{ri}}{\partial b_i} = \frac{\sum^{M}_{j=1,j\neq i}b_j+\xi}{(\sum^{M}_{j=1}b_j+\xi)^2}$. The first factor in the brackets is a strictly decreasing function of $b_i$, and $\frac{\partial P_{ri}}{\partial b_i}$ is always positive, so the payoff function is a strictly quasi-concave function of $b_i$, which indicates that  there exists at least one NE.
The unique best response   for each player can be obtained by setting $\frac{\partial U_i(b_i; \mathbf{b}_{-i},\pi)}{\partial b_i} =0$, and a desirable outcome for the power allocation game is
 \begin{eqnarray}
 P_{ri} =  \left[\frac{1}{2\ln2 \pi}-\frac{1}{|g_i|^2}\right]^+,
 \end{eqnarray}
where $\left(x\right)^+$ denotes   $\max\{x, 0\}$. By using the fact that the power that each user can get   is bounded, i.e. $0\leq P_{ri}\leq P_r$,   the first part of the proposition can be proved.

The uniqueness of NE can be proved  by studying the  contraction  mapping of the best response functions. Consider $\pi< \min \{\pi_{u,i}, i\in \mathcal{N}\} $, and define $\varrho_i(\pi)=\frac{\left( \frac{1}{2\ln 2 \pi} - \frac{1}{|g_i|^2}\right)}{P_r - \frac{1}{2\ln 2 \pi} - \frac{1}{|g_i|^2} }$. Therefore it is necessary to prove that there exists $\nu\in (0,1)$ such that for any $\mathbf{x}$ and $\mathbf{y}$ in $\mathbb{R}^+$, $||BR(\mathbf{x})-BR(\mathbf{y})||_2\leq \nu ||x-y||_2 $, where $BR(\mathbf{b})=BR_1(\mathbf{b}_{-1})\times\cdots\times BR_N(\mathbf{b}_{-N})$, the Cartesian product of the best response function of each user and $||x||_2$ denotes the norm operation.  Consider two distinct possible action sets, $\mathbf{x}$ and $\mathbf{y}$. From \eqref{best response function}, $||BR(\mathbf{x})-BR(\mathbf{y})||_2$ can be expressed as
 \begin{eqnarray}
&&\left(\sum^{N}_{i=1} \left(BR_{i}(\mathbf{x}_{-i})- BR_{i}(\mathbf{y}_{-i})\right)^2\right)^{\frac{1}{2}} \\ \nonumber &=&\left(\sum^{N}_{i=1}\varrho_i^2(\pi)\left(\sum^{N}_{j\neq i} (x_j-y_j)\right)^2\right)^{\frac{1}{2}} \\ \nonumber &=&\left(\sum^{N}_{i=1}\varrho_i^2(\pi)\left(\zeta -  (x_i-y_i)\right)^2\right)^{\frac{1}{2}},
 \end{eqnarray}
 where $\zeta=\sum^{N}_{j=1} (x_j-y_j)$. The above expression  can be bounded as
  \begin{eqnarray}
&&\left(\sum^{N}_{i=1}\varrho_i^2(\pi)\left(\zeta -  (x_i-y_i)\right)^2\right)^{\frac{1}{2}}
\\\nonumber &\leq&\left(\sum^{N}_{i=1}\varrho_i^2(\pi)\left(|\zeta| +  |x_i-y_i|\right)^2\right)^{\frac{1}{2}}
\\ \nonumber &\underset{(a)}{\leq}&\left(\sum^{N}_{i=1}\varrho_i^2(\pi) \right)^{\frac{1}{2}}\left|\sum^{N}_{j=1} (x_j-y_j)\right|\\ \nonumber &&+\left(\sum^{N}_{i=1}\varrho_i^2(\pi)\left( x_i-y_i\right)^2\right)^{\frac{1}{2}}
\\ \nonumber &\underset{(b)}{\leq}&N^{\frac{1}{2}}\left(\sum^{N}_{i=1}\varrho_i^2(\pi) \right)^{\frac{1}{2}}\left(\sum^{N}_{j=1} (x_j-y_j)^2\right)^{\frac{1}{2}}\\ \nonumber &&+\left(\sum^{N}_{i=1}\varrho_i^2(\pi)\left( x_i-y_i\right)^2\right)^{\frac{1}{2}}
\\ \nonumber &\leq&\mu\left(\sum^{N}_{j=1} (x_j-y_j)^2\right)^{\frac{1}{2}},
 \end{eqnarray}
where $|x|$ denotes the absolute value of $x$,  $\varrho_{max}(\pi)=\max\{\varrho_1(\pi), \cdots, \varrho_N(\pi)\}$, $\mu=\left(N^{\frac{1}{2}}\left(\sum^{N}_{i=1}\varrho_i^2(\pi) \right)^{\frac{1}{2}}+\varrho_{max}(\pi)\right)$, the step (a) follows from the Minkowski¡¯s inequality and the step (b) follows from the Cauchy inequality. Since $\varrho_i(\pi)$ is a decreasing function of $\pi$,   there exists a threshold such that when $\pi $ is larger than this threshold, $\mu<1$  and
   \begin{eqnarray}
\left(\sum^{N}_{i=1} \left(BR_{i}(\mathbf{x}_{-i})- BR_{i}(\mathbf{y}_{-i})\right)^2\right)^{\frac{1}{2}}<     \left(\sum^{N}_{j=1} (x_j-y_j)^2\right)^{\frac{1}{2}},
 \end{eqnarray}
which means that the best response function is a contraction mapping, and therefore there exists a unique NE \cite{Bertsekas}.  Thus the proposition is proved.
 \hspace{\fill}$\blacksquare$\newline

 \bibliographystyle{IEEEtran}
\bibliography{IEEEfull,trasfer}

\end{document}